\documentclass[10pt, journal]{IEEEtran}
\IEEEoverridecommandlockouts

\usepackage{algorithm,algpseudocode}
\algnewcommand{\Initialize}[1]{%
	\State \textbf{Initialize:}
	\Statex \hspace*{\algorithmicindent}\parbox[t]{.8\linewidth}{\raggedright #1}
}
\algnewcommand{\Inputs}[1]{%
	\State \textbf{Inputs:}
	\Statex \hspace*{\algorithmicindent}\parbox[t]{.8\linewidth}{\raggedright #1}
}

\usepackage{amsmath}
\usepackage{amssymb}
\usepackage{amsthm}
\usepackage{psfrag}
\usepackage{bm}
\usepackage{multirow}
\usepackage{enumerate}
\usepackage{balance}
\usepackage{epstopdf}
\usepackage{subfigure}
\usepackage{booktabs}
\usepackage{threeparttable}
\usepackage{epsfig}
\usepackage{color}
\usepackage{xcolor}
\usepackage{bbm}
\makeatletter

\newcommand{\Rmnum}[1]{\expandafter\@slowromancap\romannumeral #1@}

\makeatother

\usepackage[justification=centering]{caption}

\newtheorem{definition}{Definition}

\newtheorem{theorem}{Theorem}

\usepackage{url}
\usepackage{cite}
\usepackage{graphicx}
\usepackage{array}
\usepackage{multirow}

\begin{document}
	%
	\title{Peer Offloading with Delayed Feedback in Fog Networks}
	\author{\IEEEauthorblockN{Miao Yang, Hongbin Zhu, Hua Qian,~\textit{Senior Member,~IEEE}, Yevgeni Koucheryavy,~\textit{Senior Member,~IEEE},  Konstantin Samouylov, and Haifeng Wang}
		\thanks{This paper was presented in part at the 2020 IEEE International Conference on Acoustics, Speech, and Signal Processing (ICASSP 2020).
		
		M.~Yang and H.~Qian are with Shanghai Advanced Research Institute, Chinese Academy of Sciences (CAS), also with School of Information Science and Technology, ShanghaiTech University, and also with University of Chinese Academy of Sciences.
		H.~Zhu and H. Wang are with School of Information Science and Technology, ShanghaiTech University, also with Shanghai Institute of Microsystem and Information Technology, CAS.
		Yevgeni Koucheryavy is with Tampere University, Finland, and also with 
		National Research University Higher School of Economics, Moscow, Russia. 
		Konstantin Samouylov is with Peoples Friendship
		University of Russia (RUDN University), Moscow, Russian Federation.
		Corresponding author: Hongbin Zhu (zhuhb1@shanghaitech.edu.cn).			
		
		This work was supported in part by the National Natural Science Foundation of China (Grant No. 61971286) and the Science and Technology Commission Foundation of Shanghai (Grant No. 19DZ1204300).}}
	
	\maketitle
	\begin{abstract}		
		Comparing to cloud computing, fog computing performs computation and services at the edge of networks, thus relieving the computation burden of the data center and reducing the task latency of end devices.		
		Computation latency is a crucial performance metric in fog computing, especially for real-time applications.
		In this paper, we study a peer computation offloading problem for a fog network with unknown dynamics.
		In this scenario, each fog node (FN) can offload its computation tasks to neighboring FNs in a time slot manner.		
		The offloading latency, however, could not be fed back to the task dispatcher instantaneously due to the uncertainty of the processing time in peer FNs.
		Besides, peer competition occurs when different FNs offload tasks to one FN at the same time.
		To tackle the above difficulties, we model the computation offloading problem as a sequential FN selection problem with delayed information feedback.
		Using adversarial multi-arm bandit framework, we construct an online learning policy to deal with delayed information feedback.
		Different contention resolution approaches are considered to resolve peer competition.
		Performance analysis shows that the regret of the proposed algorithm, or the performance loss with suboptimal FN selections, achieves a sub-linear order, suggesting an optimal FN selection policy.
		Besides, we prove that the proposed strategy can result in a Nash equilibrium (NE) with all FNs playing the same policy. 	
		Simulation results validate the effectiveness of the proposed policy.		
	\end{abstract}

	
	\section{Introduction}
	\label{sec:intro}

	The pervasive mobile devices, together with the rapidly increasing mobile applications, such as, augmented reality, virtual reality, online gaming and super-resolution video streaming, are revolutionizing our way of life \cite{atzori2010internet}.
	These massive devices and heterogeneous services require unprecedentedly high transmission speed and low latency \cite{Elbamby2019}. 
	The limited battery capacities and computation resources of the mobile devices, on the other hand, place stringent restrictions on local computations.
	
	Fog computing has recently emerged to assist in the task processing of mobile devices \cite{Bonomi2012, Yi2015, Bonomi2014}. 
	The key idea of fog computing is to utilize computation resources of fog nodes (FNs), which are usually located at the edge of networks.
	These FNs can be the servers or micro base stations with extra computation resources. 
	The computation tasks of mobile devices can be offloaded to FNs, which releases the computation workloads of mobile devices and reduces their power consumption.	
	On the other hand, the offloaded computation tasks and the computation capability of each FN vary, which may result in unbearable latency and degrade the users' experience. 
	For example, some weak FNs with poor computation capacity or extensive workload are not in favor of computation-intensive applications.
	
	Intelligent computation peer offloading strategy is desired to fully utilize the neighboring FNs with high quality of experience.
	By means of peer computation offloading, the weak FNs can offload their emerging computation tasks to the FNs in the vicinity.
	With efficient peer offloading strategy, fog computing can assist the terminals to satisfy stringent latency requirements and support a variety of internet of things (IoTs) devices whose computation resources are limited.
	
	\subsection{Prior Work}
	In recent literatures, quite a  few computation offloading problems over the fog networks have been proposed.
    The authors in \cite{Zhang2017} investigated the computation offloading problem by adaptively offloading the computation tasks from cloud to fog servers with the predicted mobility. 
    The task offloading problem was formulated as a multi-period generalized assignment problem in \cite{Wang2019}, which aimed at minimizing the total delay by offloading tasks to suitable fog nodes in each period.
    The work in \cite{Deng2016} studied a workload assignment problem between fog and cloud that considered the trade-off between power consumption and transmission delay. 
    The above studies are based on centralized scheduling with high overhead of information exchange at edge devices.  
    Such offloading algorithms do not fit for real-time and massive data services. 
    In addition, these studies do not consider the offloading among FNs, which can be more challenging with unbalanced workload among FNs.
	
	To tackle the peer offloading problem, the authors in \cite{CY2017} and \cite{TDinh2017} formulated the offloading problem as a convex optimization problem to minimize the energy consumption with latency constraints.
	In \cite{Kety2016}, a latency minimization problem was studied to allocate the computation resources of the edge servers.
	Meanwhile, the authors in \cite{khaledi2016} proposed a task allocation approach that minimized the overall task completion time using a multidimensional auction. 	
	In the above works,  it is generally assumed that information of the fog network is completely known, thus the computation offloading is modeled as a deterministic scheduling problem.
	However, such assumptions do not hold in practice since the global information could not be obtained in real-time and the network status may change rapidly.
	
	For the task offloading problem with unknown dynamics, the authors in \cite{Zhu2019} proposed an online peer offloading algorithm for time-varying channel conditions and stochastic arrival tasks.
	The work in \cite{GL2017} proposed an online secretary algorithm to minimize the offloading latency with unknown arrival rates of neighboring FNs.
	The latency, the energy consumption, and the switching cost were jointly considered as a stochastic programming problem in \cite{ZhuZW2019}.
	Taking into account the requirements of ultra-reliable low-latency, the authors in \cite{Liu2019} proposed a mechanism to minimize the users' power consumption while considering the allocated resources.
	Besides, in \cite{MY2020}, the workloads of neighboring FNs were formulated as a Markov model and a latency minimizing algorithm was proposed to track these dynamics.
	These works assume that the workload of FNs satisfies some stationary distributions.
	Furthermore, there exist competition and collision effects among other FNs offloading choices.
	Some unexpected tasks brought by dynamic network structure and conditions could deteriorate the offloading performance.
	
	In general, all works mentioned above assume instantaneous feedback of offloading decisions.
	The task assigner is assumed to receive the expected latency when the tasks are dispatched.
	This assumption is not always valid in the real world since the FN that provides computation ability may not be able to give accurate latency estimates when the tasks are received.
	The exact execution clock cycle of each task can not be obtained until it is finished. 
	The execution time of a task is coupled with the hardware and the software state, such as, the cache condition, resource contention, scheduling policy, interrupt management, etc. 
	In fact, even the measurement of the worst-case execution time (WCET) of a task is quite a headache problem in the field of real-time operating system (RTOS) \cite{gustafsson2010malardalen}.

\subsection{Contribution and Organization}
	In this paper, for the FN who wants to offload tasks, we assume that: (1) the global FNs state information is not known; (2) the resources of service FNs may be unpredictably competed by other FNs; (3) the offload feedback is obtained only when the task was processed and finished by the computation FN.
	These assumptions naturally raise three challenging problems.
	Firstly, the performance achieved from each FN is unaware beforehand and has to be learned due to dynamic FN conditions. 
	Secondly, the achieved performance also depends on the set of competing tasks, which is unpredictable since the decisions of other FNs are made locally without coordination. 
	Thirdly, the feedback of the offloaded task is delayed, which may deteriorate the achieved performance of existing strategies.
	Most algorithms for computation offloading do not fully consider the uncertainties mentioned above and are not in favor of these assumptions.
	
	In \cite{Yang2020}, we studied the computation peer offloading problem in a fog-enabled network. 
	Our goal is to minimize the peer offloading latency in uncertain environment without instantaneous latency feedback.
	An exponential based algorithm was proposed to handle the tricky problem and simulation results revealed the efficiency of the proposed strategy.
	
	In this paper, we extend our previous work with theoretical proofs for the single-agent setting in \cite{Yang2020}.
	Furthermore, we study the computation peer offloading problem for the multi-agent setting.
	The main contributions of this paper are summarized as follows.
	\begin{itemize}
		\item We develop an online learning framework for fog networks, where each FN can offload their computation tasks to peer FNs in the vicinity. We formulate the latency minimization problem as a sequential decision making problem. Both single-agent and multi-agent settings are considered in this work.
	    \item For the single-agent setting, an online learning algorithm is advocated based on the adversarial multi-armed bandit (MAB) framework, which accommodates arbitrary variation of FN conditions.
    	The proposed method can be harnessed to learn the statistical information of the network through delayed information feedback.
    	Besides, the proposed method allows task dispatching FN to dynamically identify the most suitable service FNs without the prior knowledge of arrival tasks and channel conditions.
    	The theoretical proof is provided, suggesting that the proposed algorithm achieves a sub-linear order with delayed feedback.
    	\item For the multi-agent setting, we consider the different conflict resolutions when multiple tasks are dispatched to the same FN.
    	Assume that all FNs adopt the same strategy, multi-agent setting can be regarded as a game between an FN and the rest of FNs.
    	The proposed algorithm for the multi-agent setting can achieve the same optimality performance as the single-agent setting.
    	Besides, the proposed algorithm tailored for the offloading game can converge to an $\epsilon$-NE in a specific two users setting.
	\end{itemize}

	The rest of the paper is organized as follows. 
	The system model and optimization problem are introduced in section \ref{sec:system model}.
	An online learning algorithm for computation offloading with delayed feedback is proposed in section \ref{sec:algorithm}.
	The theoretical analysis of multi-agent games with the proposed algorithm is studied in \ref{sec:multi-user}.
	Numerical results are given in section \ref{sec:numerical results} and	section \ref{sec:conclusions} concludes our work.
	
	\section{System Model and problem formulation}
	\label{sec:system model}
	\subsection{Network Model}
	Consider a fog network consisting of an end-user layer and a fog layer as shown in Fig.~\ref{fig:system_model}.
	In our framework, the end-user layer includes smartphones, laptops and wireless sensors that do not have enough computation capability and are power limited.
	With the help of fog computing, these end-users can offload their task data to the fog layers for remote distributed computing purposes.
	We consider a fog network with $N$ FNs, defined by $\mathcal{N}=\{1,2,...,N\}$.
	Each FN can process the task from end-user devices or other peer FNs, which is typical in practical fog networking \cite{MChiang2016}.
	
	For concisely, we classify FNs into the following two categories:
	(1) \textbf{Task FN}: An FN is a task FN if it offloads workloads to other FNs. 
	Moreover, it does not receive any workload from other FNs;
	(2) \textbf{Service FN}: An FN is called a service FN if it receives workloads from other FN and does not offload workload to others.	
	Let $V$ and $K$ denote the amounts of task FNs and service FNs in fog networks.
	Notice that in our categorization, a smart FN won't offload workloads to other FNs while receiving workloads from other FNs. 
	In other words, the FN can only play a role in the task FN and the service FN.  
	
	The operational timeline is discretized into time slots (e.g., $10$ ms) for making peer offloading decisions.
	Once FN $i$ receives $b_t^i$ task at time slot $t$, the FN can process the task locally or offload the task to its neighboring FNs.
	The service rate, which denotes the expected number of consecutive task completions per time slot, can help character the computation capability of FN. 
	Intuitively, the FNs with low service rates or their central processing units (CPUs) are running at full load prefer to offload computation tasks to its neighboring FNs, while the FNs with high service rates or abundant idle computation resources prefer to process locally.
	\begin{figure}[t]
	\centering
	\includegraphics[scale=0.45]{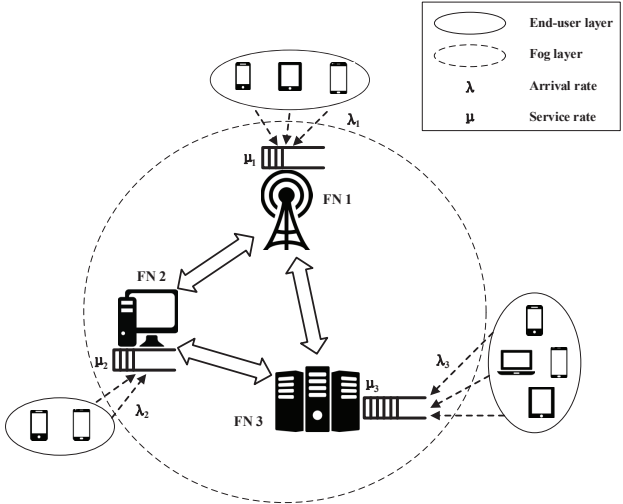}
	\caption{Peer offloading scenario.}
	\label{fig:system_model}
	\centering
	\end{figure}		
	
	Next, we discuss the latency of peer computing and local computing, which is our foremost concern. 
	The latency of peer computing contains three parts: (1) latency of task dispatching; (2) latency of task processing; (3) latency of results receiving.
	The latency of results receiving can be negligible due to the relatively small size of computation results \cite{LChen2018}.

	\subsection{Communication Model}
	Most works about computation offloading assumed that some prior information is available, such as, wireless channel state information (CSI) \cite{huang2012, Munoz2015}.
	However, it is not easy to have access to such information, especially when the number of FNs is large. 
	In this paper, CSI is assumed to be time-varying and unknown before task transmission.
	Task FN $i$ can offload tasks to service FN $j$ over a wireless channel.
	The transmitted rate can be given by
	\begin{align}
		r^{ij} = B\log_2(1+\frac{g^{ij}hP^i}{\sigma^2}),
	\end{align}
	where $g^{ij}$ is the channel gain between FN $i$ and $j$, and $h$ is the average fading gain of the FN $i$. 
	Notation $P^i$ is the transmission power of FN $i$ and $\sigma^2$ denotes the noise power spectral density.
	The bandwidth per node is the same in our model and is given by $B$.
	We denote $\overline{r}^{ij}$ as the average transmission rate between FN $i$ and FN $j$ during one task transmission.
	Thus the task dispatching latency can be described as
	\begin{align}
		D_t^i = \frac{b_t^i}{\overline{r}^{ij}}.
	\end{align}

	The channel deterioration may lead to the failure of the computation tasks' transmissions.
	The task FN can confirm the transmission status by acknowledgement (ACK) or negative-acknowledgement (NACK) message \cite{3gpp36213}.
	Once receiving a NACK message, the task FN will re-transmit the task with the Automatic Repeat-reQuest (ARQ) or the Hybrid Automatic Repeat-reQuest (HARQ) scheme \cite{costello1983error}.
	\subsection{Computation Model}
	The majority of existing offloading studies consider an ideal setting that the task FN can obtain the computation latency when tasks arrive at the service FN.
	However, current OSs do not support instantaneous feedback.
	The initial machine state of the service FN varies before the task executes and some interrupts may occur during the computation processing. 
	In this paper, the feedback latency can be obtained when the service FN finishes the execution of the offloaded task.
	
	Next, we define the computation model in the service FNs.	 
	According to Kleinrock approximation \cite{Bertsekas1992}, the arrival process of the task can be regarded as a Poisson process.
	To model the perturbation of the computation process in service FNs, the service rate of each FN is assumed to be negative exponentially distributed (i.e., generated by a Poisson process).
	Besides, the statistics of arrival rates vary among FNs since the density of end-users is usually not uniform.
	The arrival rate $\lambda$ of an FN is unaware to other FNs.
	
	When a task arrives at the destination, it will be cached in the queue if the service FN is busy.
	Since the computing service rate $\mu_t^j$ of FN $j$ at time slot $t$ may be perturbed by some system interrupts, we assume that $\mu_t^j$ is unknown to all FNs.
	Let $\bar{\mu}_t^j$ denote the average service rate during computation.
	The computation queue can be modeled as an $M/M/1$\footnote{$M$ denotes Markov property. } queue. 
	An $M/M/1$ queueing system indicates that the tasks arrive according to a Poisson process and the service rates are assumed to be independent and exponentially distributed.
	This is consistent with what we described earlier.
	Let $Q_t^j$ denote the current task queue length of FN $j$ when the task arrives at time slot $t$.
	
	Without loss of generality, we assume that the task computation latency in service FNs is proportional to the received task size, such as,  $kb_t^i, k\in R^+$.
	Thus, the computation latency at FN $j$ is
	\begin{align}
		C_t^j = \frac{Q_t^j+kb_t^i}{\bar{\mu}_t^j}.
		\label{equ:computation latency}
	\end{align}	
	Equation (\ref{equ:computation latency}) denotes an estimation of the computation latency.
	\subsection{Problem Formulation}
	Combining the communication latency and computation latency, the total peer offloading latency is
	\begin{align}
		O_t^j = C_t^j + D_t^i=\frac{Q_t^j+b_t^i}{\bar{\mu}_t^j} + \frac{b_t^i}{\overline{r}^{ij}}.
	\end{align}
	
	In practice, the offloading latency can not be infinite, we bound it within $[0, T_{\text{max}}]$, where $T_{\text{max}}$ is the maximal tolerable latency.
	Thus, the loss $l_t^j$ of each choice is described as
	\begin{align}
		l_t^j = \begin{cases}
			\frac{O_t^j}{T_{\max}}, &\text{if}~ O_t^j < T_{\text{max}},\\
			1, &\text{otherwise},
		\end{cases}
		\label{equ:normalized_reward}
	\end{align}
	where $l_{t}^{j} \in [0, 1]$. 
	We want to minimize the average latency of each FN in the long term and it can be formulated as
	\begin{equation}
		\begin{split}
			&\text{minimize}~~~ \lim_{T\rightarrow\infty}\frac{1}{T}\sum_{t=1}^{T}\sum_{i=1}^{K}L_t^i~\mathbbm{1}\{I_t=i\}\\
			&\text{subject to}~~~~~ I_t\in \mathcal{I}, t=1,2,...,T,
			\label{equ:optimization}
		\end{split}
	\end{equation}
	where $I_t$ denotes the index of the offload FN at time slot $t$, $\mathcal{I}$ denotes the set of all the offloading choices and $\mathbbm{1}$ represents indicating vector.
	In our model, the channel conditions and task queues of the service FNs are unaware to the task FN.
	The task FN could not receive the latency results until the feedback has been received.
	Finding the optimal solution of (\ref{equ:optimization}) is challenging with the non-realtime information.
	
	Online learning is a well-established paradigm for making a sequence of decisions when given the knowledge of previous prediction feedback.
	Besides, online learning can reduce computation complexity and has been investigated in many literatures.
	Integrating online learning scheme, the challenging optimization problem of (\ref{equ:optimization}) can be converted into the low-complexity sequential decision problem at each time slot.
	
	Vanilla online learning algorithm operates in a fixed order, such as making choices, receiving feedbacks and updating the learning model.
	These typical online learning algorithms may not work under the delayed feedback setting, where the order of the received feedback may not be in line with the order of the dispatches.
	The lag of the feedback and disorder of the observations bring significant challenges to vanilla online learning algorithms.
	
	\section{Single-agent setting}
	\label{sec:algorithm} 
	In this section, we first introduce the adversarial multi-arm bandit (MAB) framework, then an offloading algorithm based on the exponentially weighted strategy is proposed.
	\subsection{Adversary multi-arm bandit}
	\label{sec:adversial_bandit}
	The adversarial MAB is an efficient algorithm when engaged in the dynamic environment.
	The offloading problem in section \ref{sec:system model} can be modeled as an adversarial MAB consisting of a set of players (task FNs), which can choose an action (service FN) from the action sets (available service FNs).	
	When the task FN makes a choice at time slot $s$, the offloading loss can be observed by the task FN after $d_s$ time slots.
	The action of offloading is adopted by one FN at slot $s$ and the loss is obtained at slot $t$, so we have $s+d_s = t$.
	
	To measure the utility of the task offloading scheme $\Phi$, the concept of \emph{expected regret} is introduced.
	Expected regret is harnessed to distinguish the cumulative loss of $\Phi$ and the best strategy in hindsight.
	The desired $\Phi$ is to minimize its expected cumulative regret $R_{\Phi}(T)$, which is defined as
	\begin{align}
		R_{\Phi}(T) = \mathbb{E}\Big[\sum_{t=1}^{T}l_{s|t}^a\Big]-\min_A\sum_{t=1}^{T}l_t^A,\label{equ:expect_loss}
	\end{align}
	where $\mathbb{E}(\cdot)$ represents the expectation operation, $a\in\{1,...,K\}$ denotes the action and $A$ denotes the best policy which gives the least loss. 
	Notation $l_{s|t}^a$ represents the loss introduced by the time lag and $l_t^A$ denotes the loss of optimal policy $A$.
	
	In adversarial multi-arm bandit setting, task FN chooses service FN $j$ using the mixed strategies according to a probability distribution $\bm{p}_t=[p_t^1,p_t^2,...,p_t^K]$ \cite{Auer2002}.
	Due to the variation of computation lag, the result computation latency may be observed in a disorderly manner. 
	We then introduce an algorithm to handle the delayed feedback and make the decision for task offloading.
	
\subsection{Proposed delayed exponential based algorithm}
	Let $\mathcal{L}_t$ denote the arrival feedback set at time slot $t$ and $|\mathcal{L}_t|$ denote the size of feedback set.
	Recall that the feedback at slot $t$ includes losses incurred at slot $s$.
	Once $\mathcal{L}_t$ is revealed, the FN can get the estimated loss according to $\bm{p}_t$ at the current slot.
	For each $l_{s|t}^j\in \mathcal{L}_t$, a biased estimator $\tilde{l}_{s|t}^j$ for reducing the estimation variance of the loss vector is
	\begin{align}
		\tilde{l}_{s|t}^j=\frac{l_{s|t}^j}{p_{s_n}^j+\beta l_{s|t}^j}\mathbbm{1}_{\{I_{s}=j\}},
		\label{equ:hat_l}
	\end{align}
	where $\beta$ is the biased factor of the loss for implicit exploration\cite{neu2015}. 
		
	During a time slot $t$, the task FN may collect more than one feedback or without any feedback. 
	Once receiving the feedback, it will be buffered and then updates $\bm{p}_{t+1}$ at the beginning of slot $t+1$.
	If the offloading feedback is returned in $[0, d_{\max}]$, we regard it as an efficient offloading action.
	On the other hand, if $d_s\geq d_{\max}$, the task FN regards the action as a failure because of the intolerable latency.
	Here, the task FN updates $\bm{p}_t$ with loss $l_{s|t}=1$.
	The computation task will be re-offloaded after $d_{\max}$ according to the selection probability.
	For $l_{s|t}^i\in \mathcal{L}_t$, the task FN $i$ updates the weights and probabilities of all service FNs with
	\begin{align}
		\tilde{l}_t^j &= \sum_{n=1}^{|\mathcal{L}_t|}\tilde{l}_{s|t}^j,\label{equ:sum_loss}\\
		w_t^j &= w_{t-1}^j\exp(-\eta_t\tilde{l_t^j}),\label{equ:update_w}\\
		p_t^j&=\frac{w_{t}^j}{\sum_{k\in K}w_{t}^k},\label{equ:p}
	\end{align}
	where $\eta$ denotes the step size. 
	If $\mathcal{L}_t=\emptyset$, the task FN directly reuses the previous distribution, i.e., $\bm{p}_{t+1}=\bm{p}_t$, and the offloading strategy does not change at all.
	When the feedback is received without lag, the updating scheme of the proposed algorithm is similar to the vanilla EXP3-IX algorithm \cite{neu2015}.
	\begin{algorithm}[H]
		\caption{Delayed Exponential Based (DEB) Algorithm.}
		\label{alg:ADEB}
		\begin{algorithmic}[1]			
			\Initialize{Set $\bm{p}_0=(\frac{1}{K},...,\frac{1}{K})$ for all FN $j \in \mathcal{K}$.
			}
			\For{$t=1,2,...,T$}
			\State Offload task to service FN according to $\bm{p}_{t-1}$.
			\State Feedback collected in set $\mathcal{L}_t$.
			\For{$n=1,2,...,|\mathcal{L}_t|$} 
			\State Estimate $l_{s|t}^j$ via (\ref{equ:hat_l}) if $l_{s|t}^j\in \mathcal{L}_t$.	
			\State Collect the loss of each FN by (\ref{equ:sum_loss}).					
			\EndFor
			\If{$d_s>=d_{\max}$}
			\State $l_s = 1$.
			\EndIf			
			\State Update $w_t^j$ and $p_t^j$ by (\ref{equ:update_w}) and (\ref{equ:p}). 
			\EndFor
		\end{algorithmic}
	\end{algorithm}
	\textbf{Algorithm \ref{alg:ADEB}} summarizes the proposed DEB algorithm for processing the task offloading problem with delayed feedback.
	Initially, the probability of each FN is set as $1/K$.
	The task FN makes the offloading choice according to the probability distribution $\bm{p}_{t-1}$.
	If fresh feedback arrives, the task FN calculates the estimated loss and updates the weight of each service FN. 
	If not, the task FN makes the decision according to the previous distribution.	
	\subsection{Performance analysis}
	In this part, we analyze the regret of the proposed delayed exponential based algorithm.
	
	\begin{theorem}
		\label{theorem:1}
		Let $T$ denote the time horizon of the learning procedure and $D=\sum_{t=1}^{T}d_t$ denotes the total delay. 
		Define $|\mathcal{M}|$ to be the regret caused by all the samples that are not received before round $T$.
		The regret $|\mathcal{M}|$ is bounded due to the intolerable feedback delay $d_{\max}$.
		With probability at least $1-\delta$, the regret of the proposed algorithm satisfies
		\begin{align*}
			\bar{\mathcal{R}}&\leq\frac{\ln K}{\eta}+\frac{\ln(K/\delta)}{2\beta}+(\frac{\eta e}{2}+\beta)(KT+\frac{K\ln(K/\delta)}{2\beta})\\
			&+D\eta +|\mathcal{M}|.
		\end{align*}
		In particular, setting $\eta=2\beta=\sqrt{\frac{\ln K+\ln(K/\delta)}{D+(e+1)KT/2}}$,
		we have
		\begin{align}
		\bar{\mathcal{R}}\leq&2\sqrt{\Big(D+\frac{(e+1)}{2}KT\Big)\Big(2\ln K-\ln({\delta})\Big)}\nonumber\\
		&+\frac{(e+1)}{2}K\ln(\frac{K}{\delta})+|\mathcal{M}|.
		\end{align}
	\end{theorem}
	\begin{proof}[Proof of Theorem 1]
		Given in Appendix A.
	\end{proof}
	According to \textbf{Theorem 1}, the regret of the proposed algorithm can be upper bounded by sub-linear order with determining proper update parameters.
	This result indicates that the regret growth decreases as time goes by, yielding the optimal service FN selections of the proposed algorithm.

	\section{Multi-agent setting}
	\label{sec:multi-user}
	In most cases, the fog network consists of multiple task FNs.
	In this section, we extend our work to the multi-agent setting.
	We consider a game where all players play according to the proposed algorithm with delayed feedback.
	Without the consideration of the delays, an algorithm with sublinear regret, played against itself, will converge to a Nash equilibrium (NE) in the ergodic average sense \cite{Cai2011}.
	In this section, we further investigate our proposed algorithm in the multi-agent setting.
	Adopting the proposed algorithm, the action of each task FN with the delay is proved to achieve an NE.
	
	\subsection{Problem Formulation}
	In the multi-agent setting, there are $V$ active task FNs and $K$ independent service FNs.
	In our paper, a time slot is divided into sub-slots.
	The $V$ task FNs are allocated to access $K$ service FNs in each sub-slots.
	In most cases, this assumption holds since not all task FNs are active to offload tasks at the same time.
	At each time slot, each active task FN chooses one service FN to offload its computation tasks based on its local observations.
	In addition, task FNs do not share information with each other due to the consideration of information security.
	In our framework, the different task FNs access service FNs in a time-division manner. 
	Collisions occur when multiple task FNs choose the same service FN.
	Two collision models are considered in our paper.
	One is the `winner-takes-all' model when a collision occurs. 
	In this model, the computation task of one task FN is processed and the rest would not be processed.
	The other model is that the service FN accepts all the arrival tasks and arranges the arrivals in a random manner in the queue without jumping.
	In the multiple task FNs setting, the total reward under a policy $\Phi$ by time $t$ is the same as (\ref{equ:expect_loss}), which is given by
	\begin{align}
	R_{\Phi}(T) = \mathbb{E}\Big[\sum_{v=1}^{V}\sum_{t=1}^{T}l_{s|t}^a(v)\Big]-\min_A\sum_{v=1}^{V}\sum_{t=1}^{T}l_t^A(v),\label{equ:expect_loss_multi}
	\end{align}
	where $l_{s|t}^a(v)$ denotes the loss of task FN $v$.
	The regret derivation procedure of the proposed algorithm in multiple agents setting is the same as the single case and here we care more about the NE property of the proposed strategy.
	
	An NE is a set of strategies, one for each of multiple players of a game, that each player's choice is his best response to the choices of the other players \cite{Nash1951}.
	Besides, the NE has the self-stability property such that the users at the equilibrium can achieve a mutually satisfactory solution and no user has the incentive to deviate. 
	In the multi-agent decentralized offloading scenario, an NE means no task FN has an incentive of changing its mixed strategy if all other FNs do not change their strategies.
	This property is crucial to the decentralized computation offloading problem, since the task FNs may act in their own interests.
	For each task FN, achieving an NE means the current strategy is their optimal response, thus yielding a minimal latency solution.

	In the following discussion, we will concentrate on investigating the NE property of the proposed algorithm. 
	We investigate the NE property of the proposed algorithm in a two users zero-sum case.
	\subsection{Equilibrium Formulation}	
	Let $U$ be the cost matrix, such that when the row player plays $i$ and the column player plays $j$, the former pays a cost of $U(i,j)$ and the latter gains a reward of $U(i,j)$ (a.k.a., a cost of $-U(i,j)$), where $0\leq U(i,j)\leq 1$ for any $i,j$.
	Notation $\bm{p}_t,\bm{q}_t$ represent the mixed strategies of row player and column player, then we use the convention that
	\begin{align}
		U(\bm{p}_t,j) \triangleq \sum_{i=1}^{K}p_t^{(i)}U(i,j),
	\end{align}
	and
	\begin{align}
		U(\bm{p}_t,\bm{q}_t) \triangleq \sum_{i=1}^{K}\sum_{j=1}^{K}p_t^{(i)}q_t^{(j)}U(i,j).
	\end{align}
	
	NE is a critical concept in game theory to predict the outcome of a game.
	An NE is a strategy profile $(\bm{p}_t^*,\bm{q}_t^*)$ such that there is no motivation for any player to change its strategy since the changed strategy will degrade the benefit. 
	To obtain the desired strategy, we need to define the set of all approximate NE:
	\begin{definition}
		The set of all $\epsilon$-NE points is
		\begin{align}
			\mathcal{N_{\epsilon}}=\Bigl\{(\bm{p}^*,\bm{q}^*)|U(\bm{p}^*,\bm{q}^*)\leq \min_{\bm{p}} U(\bm{p},\bm{q})+\epsilon,\nonumber\\ U(\bm{p}^*,\bm{q}^*)\geq \max_{\bm{q}} U(\bm{p},\bm{q})-\epsilon\Bigr\},
		\end{align}
		and the set of NE points is $\mathcal{N}_0$.
	\end{definition}
	The entity that converges to the set of NE in our case is the ergodic average of $(\bm{p}_t,\bm{q}_t)$.
	\begin{definition}
		The ergodic average of a sequence of distributions $\bm{p_t}$ is defined as:
		\begin{align}
			\bm{\bar{p}}_t\triangleq\frac{\sum_{\tau=1}^{t}\bm{p}_{\tau}}{t}.
		\end{align}
		We say that $(\bm{\bar{p}}_T,\bm{\bar{q}}_T)$ converges in $L1$ to the set of NE if
		\begin{align}
			\lim_{T\rightarrow \infty}\mathop{\arg\min}\limits_{(\bm{p}_T^*,\bm{q}_T^*)\in\mathcal{N}_0} E\{||(\bm{\bar{p}}_T,\bm{\bar{q}}_T)-(\bm{p}_T^*,\bm{q}_T^*)||_1\}=0,
		\end{align}
		which also implies that $\forall \epsilon>0$,
		\begin{align}
			\lim_{T\rightarrow \infty}\mathop{\arg\min}\limits_{(\bm{p}_T^*,\bm{q}_T^*)\in\mathcal{N}_0} P\{||(\bm{\bar{p}}_T,\bm{\bar{q}}_T)-(\bm{p}_T^*,\bm{q}_T^*)||_1\geq\epsilon\}=0,
		\end{align}
		where $||\cdot||_1$ represents $L1$ norm.
	\end{definition}

	The main theory analysis, given in \textbf{Theorem 2}, generalizes this statement for the case with delayed feedback, and the proposed algorithm can converge to an NE with the feedback lag.
	Note that the convergence of the ergodic mean to the set of NE is in the $L1$ norm sense, which is much stronger than the convergence of the expected ergodic mean.
	\begin{theorem}
		\label{theorem_equ}
		Let two players play a zero-sum game with a cost matrix $U$, where $0\leq U(i,j)\leq 1$ for each $i,j$ by the proposed algorithm.
		Denote the delay sequences of the row player and the column player as $\{d_t^r\}$ and $\{d_t^c\}$.
		Let the mixed strategies of the row and column players at round $t$ be $p_t$ and $q_t$, respectively. 
		Then as $T\rightarrow\infty$, the game played by the proposed algorithm versus itself converges to NE.
	\end{theorem}
	\begin{proof}[proof of theorem 2]
		Given in Appendix \ref{proof:nash}.
	\end{proof}
	Regarding the action of the adversary as the multiple merged actions of other players, the two players setting can be generalized to the multiple players setting \cite{roughgarden2010algorithmic}.

	\section{Numerical Results}
	\label{sec:numerical results}
	In this section, we evaluate the experiment performance of the proposed DEB algorithm.
	During all the simulations, the task size of each FN is equal and is normalized as $1$. 
	The duration of each time slot is also normalized as $1$.
	\subsection{Single-agent Setting}
	\begin{figure*}[]
		\centering
		\subfigure[Regret performance]{
			\includegraphics[width=58mm]{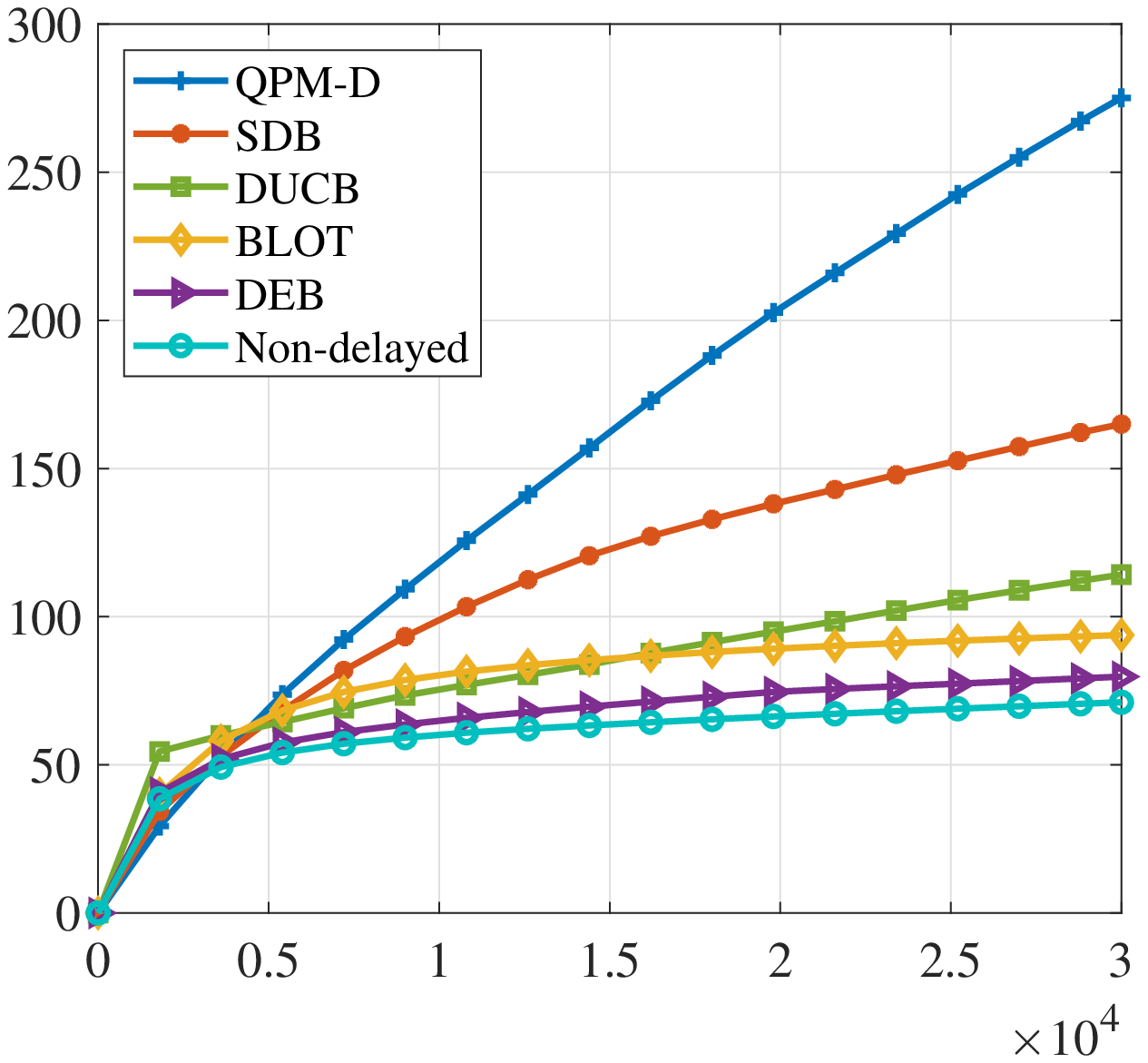}}	
		\label{fig:regret}	
		\subfigure[Latency performance]{
			\includegraphics[width=58mm]{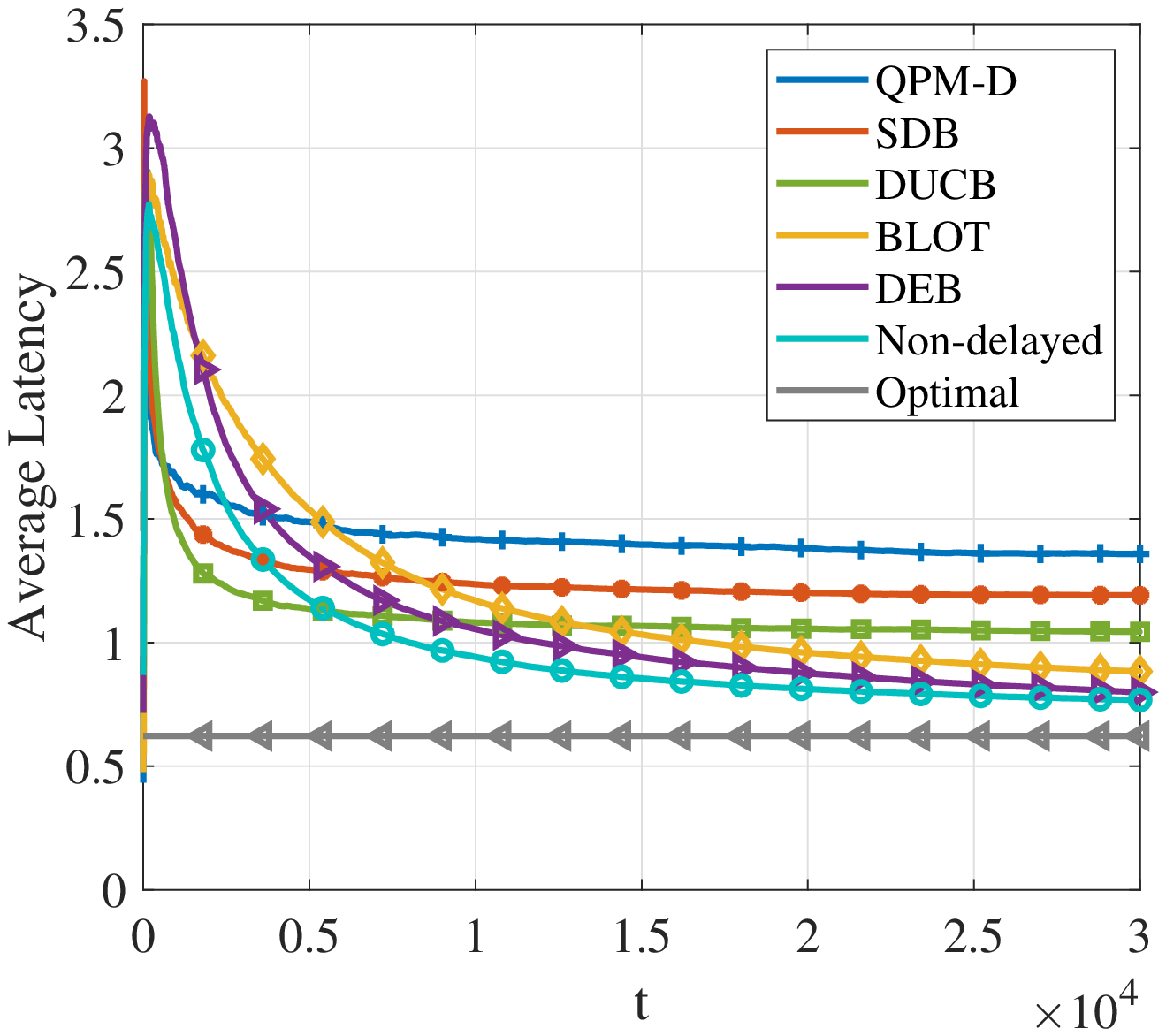}}
		\label{fig:single_latency}
		\subfigure[Statistical selection]{
			\includegraphics[width=58mm]{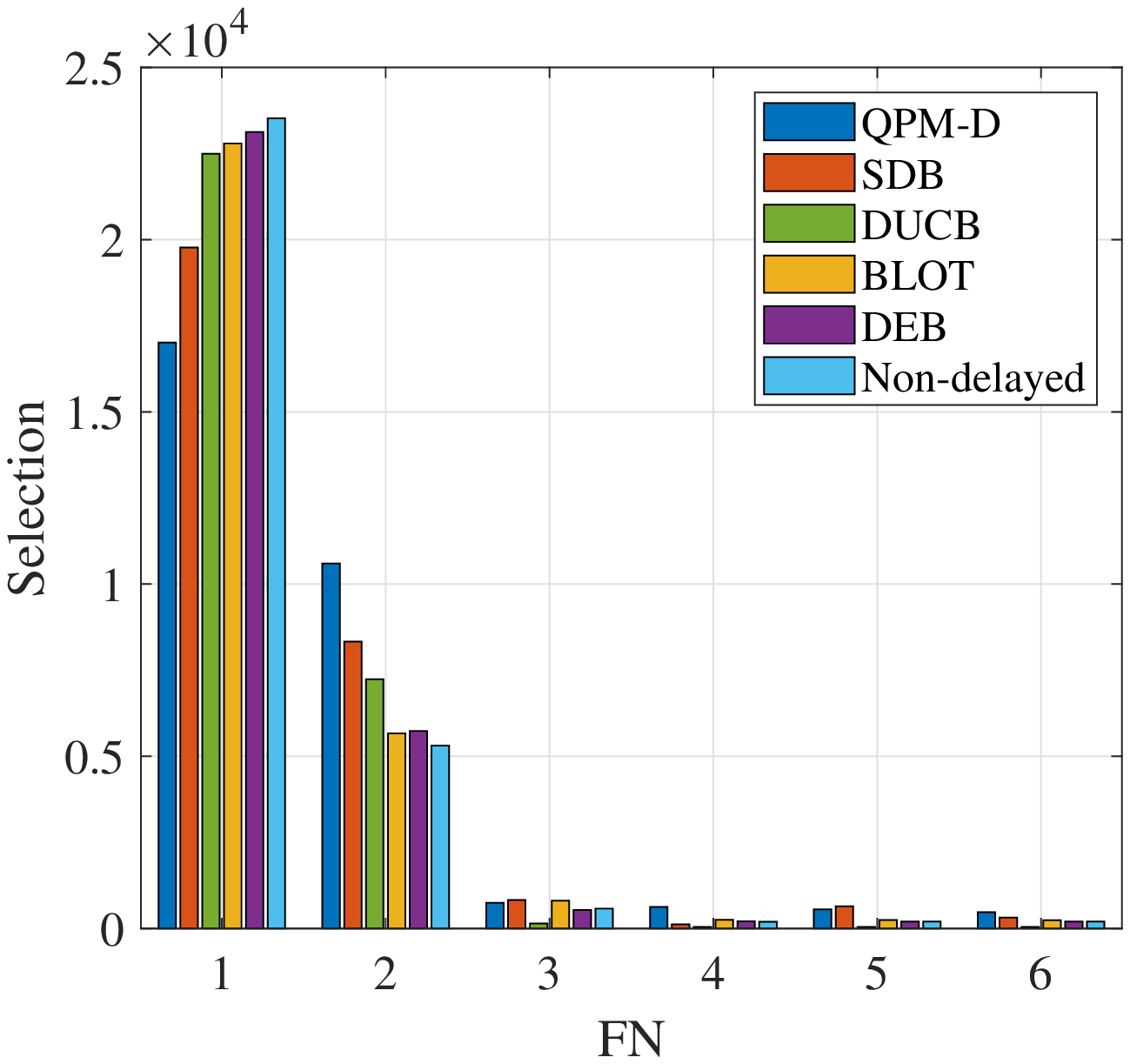}}
		\label{fig:selection}
		\centering
		\caption{Performance comparison of the QPM-D, SGD, DUCB, BLOT, DEB and no-delayed EXP3 algorithms for the single-agent setting in stationary environment.}
		\label{fig:total_single}
	\end{figure*}
	In the first part, we would like to exam the performance of the DEB algorithm in the single-agent setting.
	Consider a fog network consisting of $7$ FNs.
	The number of task FN is $1$, and the rest are service FNs.
	The task FN is deployed in the center of a circle with $2$ km radius, where the service FNs are randomly deployed in this range.
	In particular, the path-loss channel is generated according to the 3GPP propagation environment \cite{3gpp36213}.	
	The total time slot $T$ is set as $30000$.
	Besides, to perform the dynamic computation perturbation, the service rates $\mu$ of FNs follows the Poisson distribution with mean generated from the Gaussian distribution with $\mathcal{N}(6,1)$.
	To normalized the offloading latency, we set $T_{\max}$ as $5$.
	The threshold $t_{\max}$ that prevents the delay of feedback from being too large is set as $3$.

	In fog computing, to the best of our knowledge, there does not exist any literature considering delayed information feedback.
	To evaluate the performance of the proposed algorithm, we compare the performance of the proposed DEB algorithm to the following schemes that can be implemented in the delayed feedback setting. (i) Queued Partial Monitoring with Delayed Feedback (QPM-D) algorithm \cite{Joulani2013}: the feedback of each offloading decision is assigned to the queue corresponding to the selected service FN and the feedback information is employed when the service FN is chosen; (ii) Stochastic Delayed Bandits (SDB) algorithm \cite{Mandel2015}: an algorithm that combines the QPM-D for the main updated scheme and a heuristic scheme for controlling the sampling distribution to improve performance in the presence of delay;
	(iii) Delayed Upper Confidence Bound (DUCB) algorithm \cite{Vernade2017}: an algorithm that exploits the delayed feedback to adjust the reward update in classical UCB framework;
	(iv) Bandit Learning-based Offloading of Tasks (BLOT) algorithm \cite{ZhuZW2019}: an offloading algorithm that selects the service FNs by means of the currently available information;
	(v) Non-delayed offloading: each task FN receives the feedback instantly using the Exponential-weight algorithm for Exploration and Exploitation (EXP3) algorithm presented in \cite{Auer2002}.
	Here QPM-D and SDB are meta schemes for handling delayed feedback. 
	Besides, we adopt Thompson sampling algorithm for the implementation of the QPM-D and SDB algorithm.
	Note that the BLOT algorithm is not designed for the latency minimization problem, a little bit modification is taken to make BLOT work in this system setting.

	In the first experiment, we consider a stationary setting in which the arrival rates of task FNs are constant over time.
	The arrival rates of FN $i$ is set to be $\lambda_i = 4.5 + 0.5\times i$ , suggesting that the FN $1$ achieves the best delay performance in statistics.
	In Fig.~\ref{fig:total_single}, we plot the regret performance, latency performance and the task FN selections of the DEB, QPM-D, SDB, DUCB, BLOT and non-delayed EXP3 algorithms.
	In addition, the latency of optimal choice (task FN always selects service FN 1) has been plotted to show the performance gap of the above algorithms.
	Here $1000$ Monte Carlo simulations are performed.
	
	Fig.~\ref{fig:total_single} (a) depicts the regret performance of offloading decisions with the above algorithms.	
	The $y$-label of the figure represents the ratio of the aggregate regret, revealing the regret order of the DEB, QPM-D, SDB, DUCB, BLOT and non-delayed EXP3 algorithms.
	The convergence performance of the proposed algorithm represents the learning process of the environment, such as the channel condition and the computation capacities of neighboring FNs.
	From Fig.~\ref{fig:total_single} (a), we observe that the DEB algorithm achieves better regret performance than the QPM-D, DUCB, BLOT and SDB algorithms.	
	The discrepancy of the performance between DEB, QPM-D, SDB algorithms is caused by the utilization of delayed feedbacks.
	The proposed DEB algorithm harnesses the delayed feedbacks instantly.
	The QPM-D and SDB algorithms store the delayed feedbacks of a service FN to its queue firstly and update when the service FN is selected.
	Due to the delayed feedback, such updating schemes will degrade the regret performance.
	As for the performance gap between the DUCB, the BLOT and the DEB algorithms, the DUCB and BLOT algorithms focus on the recent feedback but loses sight of the entire information, which results in a larger performance loss.
	Without the knowledge of other FNs' conditions and peer competitions, the FN adopted the DEB algorithm can always adapt to the environment and make a proper choice.
	
	Fig.~\ref{fig:total_single} (b) indicates the average offloading latency of the single-agent setting.
	The offloading latency of the QPM-D and SDB algorithms decrease fast at the beginning and achieves sub-optimal latency performance later.
	Moreover, the converged latency of algorithms DUCB and BLOT perform not as well as the DEB algorithm.
	The reason comes from that the DUCB and BLOT algorithms adopt the discount UCB style to handle the delayed feedback.
	Thus the exploration procedure pay more attention to the recent feedback. Thus the two algorithms are not sufficient to identify the best service FN.
	The proposed algorithm can fully explore the service FNs and then pick the best FN, thus revealing a better averaged latency performance.
	This figure can also be leveraged to complementary explain the results in Fig.~\ref{fig:total_single} (a).
	
	\begin{figure}[t]
		\centering
		\includegraphics[width=80mm]{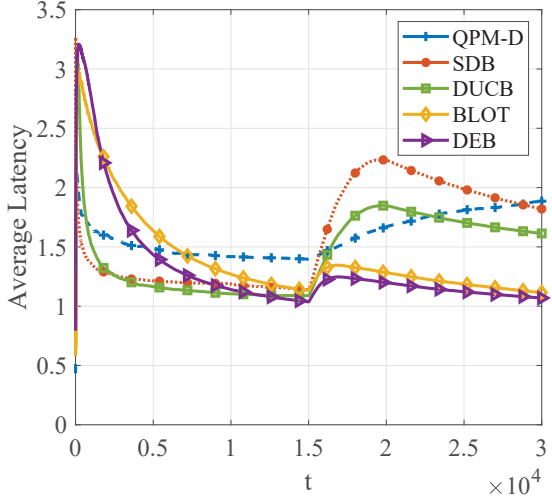}
		\caption{Performance comparison of the QPM-D, SDB, DUCB, BLOT and DEB algorithms in dynamic environment.}
		\centering
		\label{fig:total_single_d}
	\end{figure}
	In Fig.~\ref{fig:total_single} (c), the statistical selections of DEB, QPM-D, SDB, DUCB, BLOT and EXP3 algorithms are depicted with the same simulation setup.
	According to the simulation setup, FN $1$ has the smallest expected loss and FN $6$ has the largest expected loss.
	Fig.~\ref{fig:total_single} (c) reveals that the task FN selects the best FN for the most times with the DEB algorithm.
	For the FN with the largest loss, the task FN rarely offloads tasks to it and its selections occur during the exploration period for most of the time.
	Combining the three figures of Fig.~\ref{fig:total_single}, we can conclude that the DEB algorithm can better adapt to the changeable environment compared to other algorithms, thus yielding the smallest regret.

	In the second experiment,  the performance of the proposed algorithm in dynamic changing network environment is verified.
	In the first $T/2$ time slots, the arrival rates of FN $i$ is set as $\lambda_i = 4.5 + 0.5\times i$, which is the same as the first experiment.
	In the last $T/2$ time slots, the arrival rate of FN $1$ slows to $7$, then the FN $2$ starts to outperform FN $1$ and eventually becomes the leader.
	The above setting is more practical where the arrival rates of FNs may change due to the dynamics of the network structure.
	This experiment can exam the robustness of the proposed algorithm.

	In Fig.~\ref{fig:total_single_d}, we explore the average latency performance of the DEB, QPM, DUCB, BLOT and SDB algorithms in the dynamic setting. 
	Here $1000$ Monte Carlo simulations are performed.
	From Fig.~\ref{fig:total_single_d}, we observe that the performance tendencies of all the algorithms have changed at $T/2$ as the network statistic varies.
	Then the deteriorated performance  recovers eventually in the rest time.
	Contrasting with QPM-D, SDB, DUCB and BLOT algorithms, the DEB algorithm recovers faster when task FN suffers network statistic variation, suggesting a more robust property.
	It is because the DEB algorithm selects the service FNs according to a probability distribution, and the implicit exploration factor $\beta$ can help adapt to the dynamic environment.

	\begin{figure}[]
		\centering
		\includegraphics[width=84.5mm]{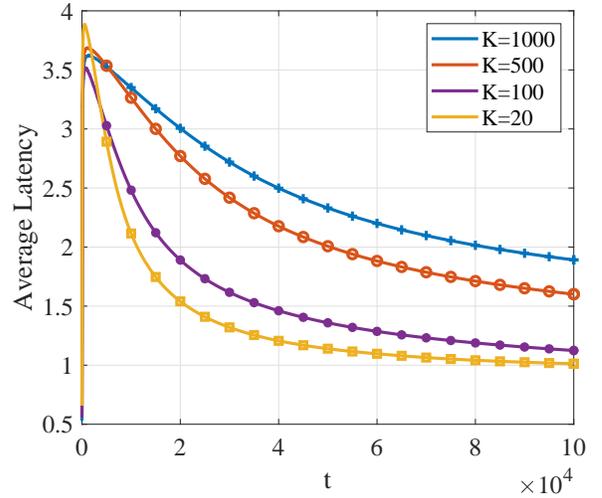}
		
		\caption{Latency performance of the DEB algorithm in the multi-agent setting with different number of service FNs.}
		\label{fig:multi_scalability}
	\end{figure}
	To evaluate the scalability of the DEB algorithm, we exam the latency performance of the DEB algorithm with different network sizes in Fig.~\ref{fig:multi_scalability}.
	Here the number of service FNs are set as $K=\{20,100,500,1000\}$.
	The FN $1$ is set as the best service FN with arrival rate $5$, while the rest service FNs are uniformly generated from $[6, 10]$.
	From Fig.~\ref{fig:multi_scalability}, we observe that the latency performance deteriorates with the increment of task FNs.
	Additional service FNs add the cost of exploration, which results in worse average latency.
	However, as time goes by, the average latency decreases. 
	The task FN with the DEB algorithm can pick the best service FN eventually.

	\subsection{Multi-agent Setting}
	In this part, we investigate the performance of the proposed algorithm in the multi-agent setting. 
	We set $V=4$ and $K=10$.
	The arrival rate of each FN $i$ is set as $\lambda_i = 4 + 0.4\times i$.
	In the multi-agent scenario, the DUCB and BLOT algorithms cause severe collisions in its exploration stage, thus the two algorithms are not applicable to multi-agent setting.
	Thus, we only compare the performance of the DEB, QPM and SDB algorithms here.
	\begin{figure}[t]
		\centering		
		\includegraphics[width=80mm]{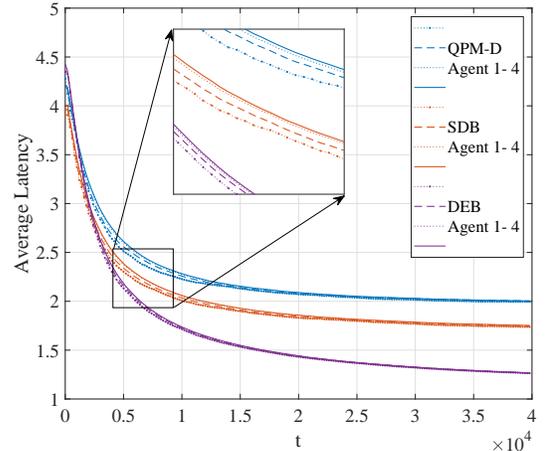}
		\label{fig:multi_latency_d}		
		\caption{Performance comparison of the QPM-D, SDB and DEB algorithms in the multi-agent setting.}
		\label{fig:multi_d}
	\end{figure}
	
	\begin{figure}[t]
		\centering
		\includegraphics[width=80mm]{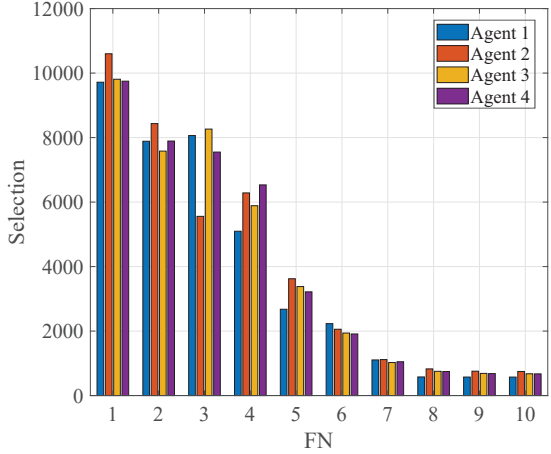}
		\caption{Service FNs selection of the DEB algorithms in the multi-agent setting.}
		\label{fig:multi_DEB}
	\end{figure}

	Fig.~\ref{fig:multi_d} explores the multi-agent performance of the DEB, SDB and QPM algorithms.
	Here we adopt a weak-collision setting, where the service FN processes the multiple arrivals in a random order manner.
	In this setting, the task FNs suffer slight collision loss when there exist multiple arrival tasks in service FN.
	Fig.~\ref{fig:multi_d} plots the average latency of different users with the above algorithms.
	Compared to the SDB and QPM algorithm, the DEB algorithm performs worse in the beginning, then gets better as time goes by. 
	The DEB algorithm chooses service FNs according to the probability distribution, indicating a sufficient exploration for all service FNs.
	Such selection rule is more suitable for dynamic Markov environment.	
	In addition, Fig.~\ref{fig:multi_DEB} shows the service FNs selections of the DEB algorithm for different agents.
	From Fig.~\ref{fig:multi_DEB}, we observe that the service FNs with low arrival rates are nearly equally selected by the four agents.	
	In this setting, the loss caused by collisions is slight.
	Therefore, the agents prefer to share the best FNs when the offloading strategy becomes stable.

	\begin{figure}[]
		\centering
		\subfigure[]{
			\includegraphics[width=42mm]{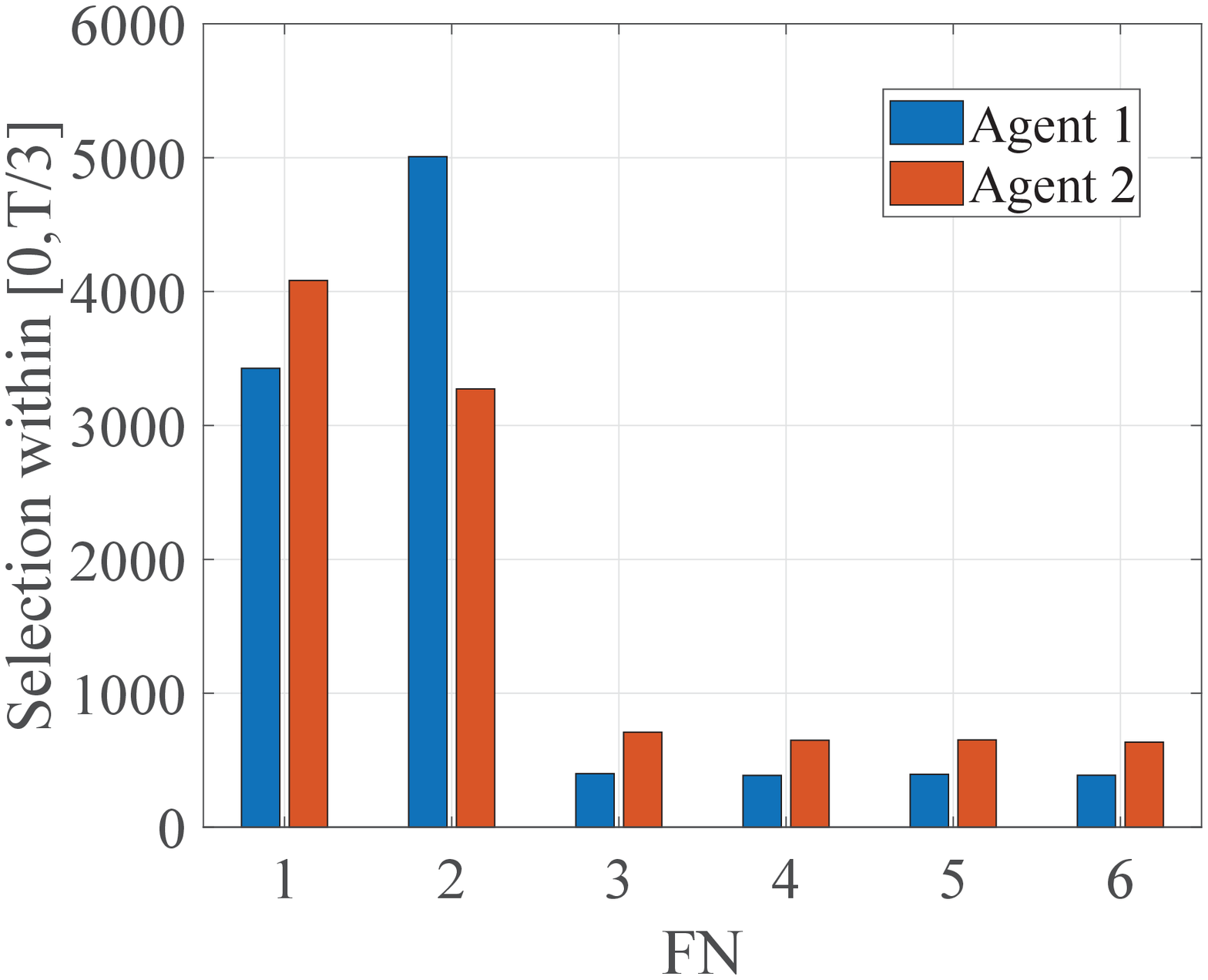}}
		\subfigure[]{
			\includegraphics[width=42mm]{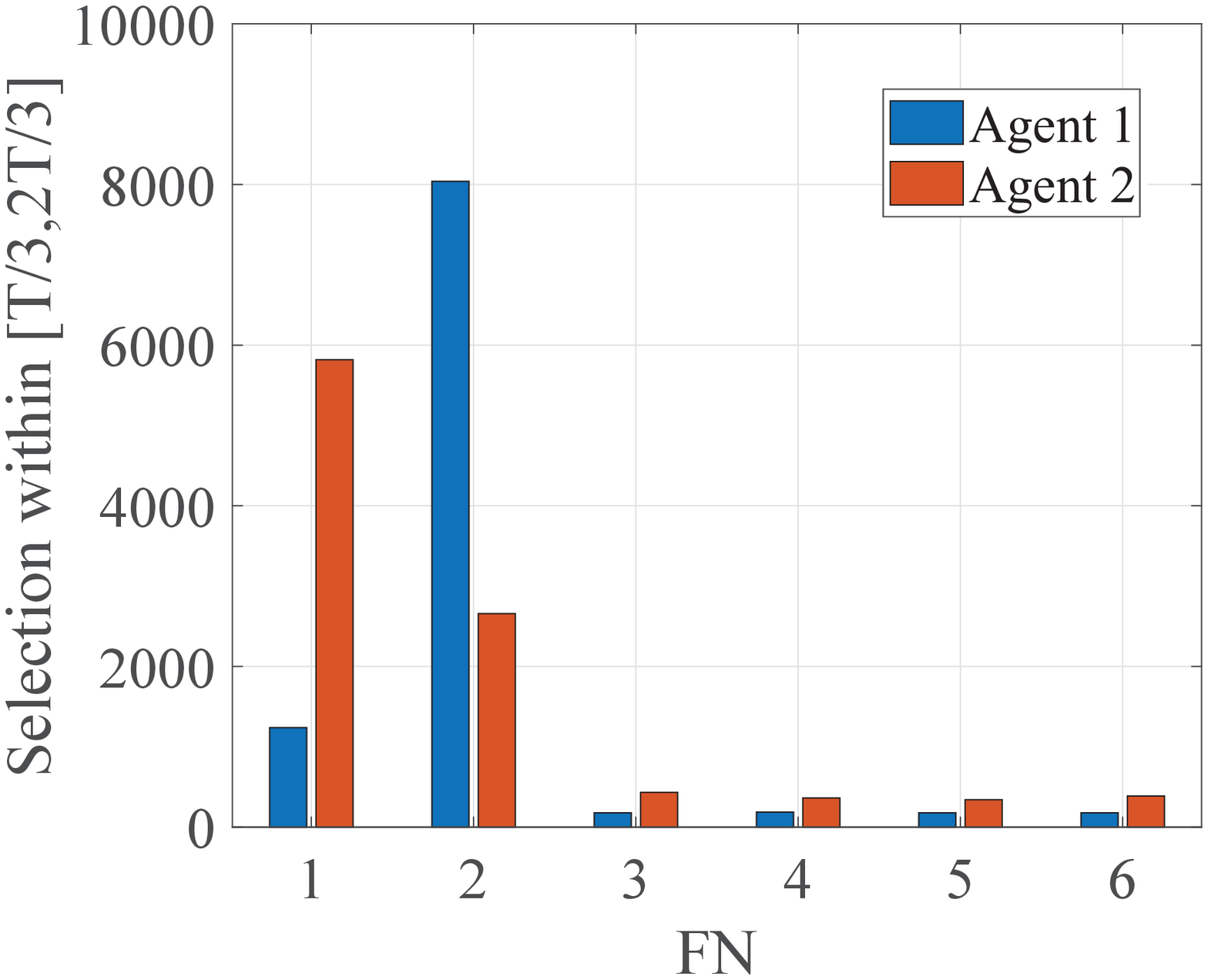}}	
		\subfigure[]{
			\includegraphics[width=42mm]{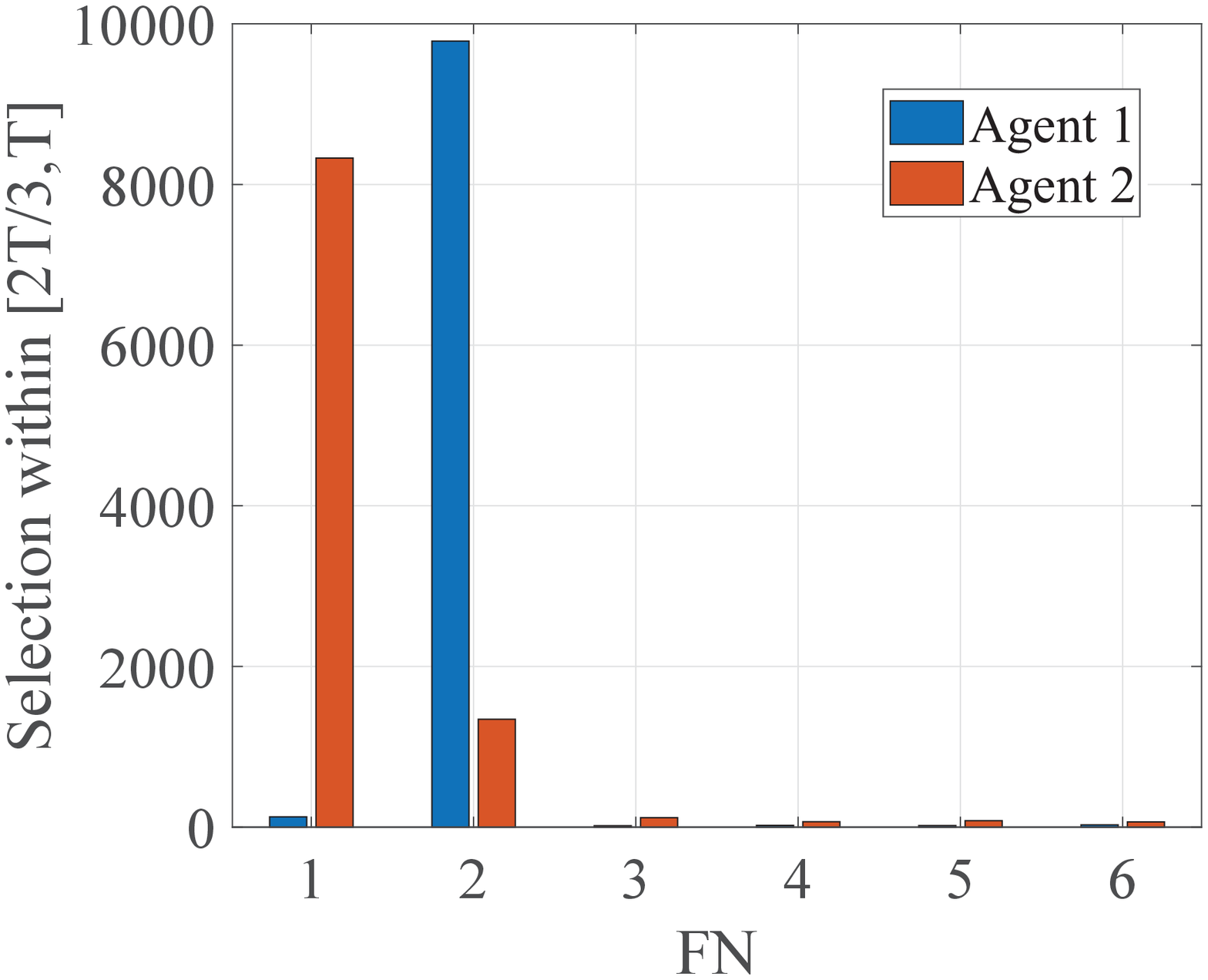}}	
		\subfigure[]{
			\includegraphics[width=42mm]{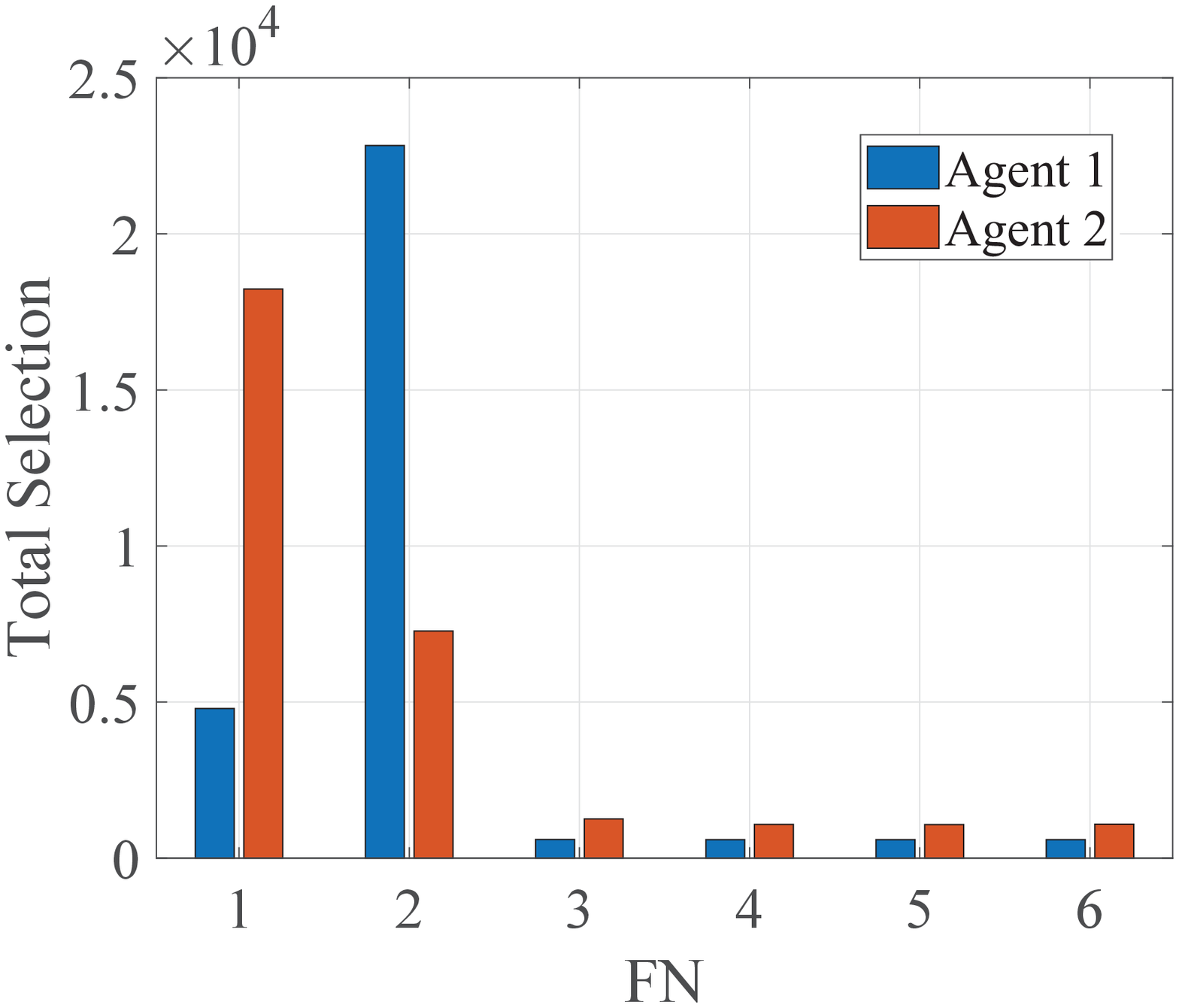}}	
		\caption{Service FNs selections in different time slot under `winner-tasks-all' setting.}
		\label{fig:NE}
	\end{figure}
	
	In the next experiment, we investigate the performance of the DEB algorithm in a `winner-takes-all' setting, where only one task FN can obtain the real loss and the rest suffer the total loss $1$.
	To reveal the NE property of the proposed algorithm, here we only consider $2$ agents.
	The number of service FNs is set as $6$.
	The arrival rate of service FN $i$ is set as $[4,4,5,6,7,8]$.
	It can be observed that the NE points in this setting is: (i) task FN $1$ chooses service FN $1$ while task FN $2$ chooses service FN $2$; (ii) task FN $1$ chooses service FN $2$ while task FN $2$ chooses FN $1$.
	Fig.~\ref{fig:NE} reveals the task FNs' selections during period [0,$\frac{T}{3}$], ($\frac{T}{3}$,$\frac{2T}{3}$], ($\frac{2T}{3}$,$T$] and $[0, T]$, respectively.
	In the earlier stage of offloading, the two task FNs prefer to share service FN $1$ and FN $2$.
	However, it will result in many collisions when their choices are the same.
	After suffering the performance loss from collisions, the two agents become smarter and gradually avoid selecting the same FN.
	Eventually, both the two task FNs achieve the NE point, where task FN $1$ choosing service FN $2$ and task FN $2$ choosing service FN $1$.
	This simulation result indicates that even with delayed feedback, the proposed algorithm can converge to an NE in the ergodic average sense.
	
	\section{Conclusions}
	\label{sec:conclusions}
    Considering the unknown peer FNs' stations and dynamic environmental changes, we studied the task peer offloading problem with minimal latency in fog networks in this paper.
	We modeled the task offloading problem as an adversarial multi-arm bandit problem with delayed feedback.
	To solve this problem, we proposed an efficient DEB algorithm that extends the general exponential based algorithm to the case with delayed feedback.
	The DEB algorithm is proven to achieve a sub-linear regret in both single-agent setting and multi-agent setting.
	Moreover, the DEB algorithm can lead to an NE with existing delay feedback in the multi-agent setting.
	Theoretical and numerical results validated the efficiency of the DEB algorithm.
	
	\appendices
	\section{}
	\begin{proof}[Proof of Theorem 1]
		Since the initial weight of each action is set as $1$, for all sequences of actions drawn by the proposed algorithm, we have
		\begin{align*}
			\frac{W_T}{W_0}
			&\geq\frac{\max\limits_j\exp(-\eta\sum_{s=1}^{T}\tilde{l}_s^j)}{K}\\
			&\geq\frac{\exp(-\eta\min\limits_j\sum_{s=1}^{T}\tilde{l}_s^j)}{K}.
		\end{align*}
		Then taking logarithms,
		\begin{align}
			\ln\frac{W_T}{W_0}\geq-\eta\min_j\sum_{s=1}^{T}\tilde{l}_s^j-\ln K. \label{equ:lower_bound}
		\end{align}
		On the other hand,
		\begin{align}
		\frac{W_t}{W_{t-1}}&=\sum_{i}\frac{w_{t-1}^i}{W_{t-1}}\prod_{s:s+d_s=t}\exp(-\eta\tilde{l}_s^i)\nonumber\\
		&=\sum_ip_t^a\prod_{s:s+d_s=t}\exp(-\eta\tilde{l}_s^i)\nonumber\\
		&\leq\sum_ip_t^a\sum_{s:s+d_s=t}\exp(-\eta\tilde{l}_s^i)\label{equ:wt_2}\\
		&\leq\sum_ip_t^a\sum_{s:s+d_s=t}\Big(1-\eta\tilde{l}_s^i+\frac{\eta^2}{2}(\tilde{l}_s^i)^2\Big)\label{equ:wt_3}\\
		&=\sum_{s:s+d_s=t}\Big(1-\eta\sum_ip_t^i\tilde{l}_s^i+\frac{\eta^2}{2}\sum_i p_t^i(\tilde{l}_s^i)^2\Big)\nonumber\\
		&=1+\Big|\{s:s+d_s=t\}\Big|-1-\eta\sum_{s:s+d_s=t}\sum_i p_t^i\tilde{l}_s^i\nonumber\\
		&~~~+\frac{\eta^2}{2}\sum_{s:s+d_s=t}\sum_i p_t^i(\tilde{l}_s^i)^2\nonumber\\
		&\leq \exp\Big(\Big|\{s:s+d_s=t\}\Big|-1-\eta\sum_{s:s+d_s=t}\sum_i p_t^i\tilde{l}_s^i\nonumber\\
		&~~~+\frac{\eta^2}{2}\sum_{s:s+d_s=t}\sum_i p_t^i(\tilde{l}_s^i)^2\Big)\label{equ:wt_4}.
		\end{align}
		From $\exp(-\eta \tilde{l}_s^i) \in (0, 1]$, $\forall i$, we can obtain (\ref{equ:wt_2}), the inequality of (\ref{equ:wt_3}) is given by $e^x\leq 1+x+x^2/2$ for $x\leq 0$ and inequality  (\ref{equ:wt_4}) is derived from  $e^x\geq 1+x$, $\forall x$.
		
		By a telescoping sum, we have
		\begin{align}
			\frac{W_T}{W_0}&\leq\exp\Big(-\eta\sum_t\sum_{s:s+d_s=t}\sum_i p_t^i\tilde{l}_s^i\nonumber\\
			&~~~+\frac{\eta^2}{2}\sum_t\sum_{s:s+d_s=t}\sum_i p_t^i(\tilde{l}_s^i)^2\Big).
		\end{align}
		Note that the sums of the form $\sum_t\sum_{s:s+d_s=t}$ only include each value of $s$ once and there may exist some feedbacks that can not be obtained before round $T$.		
		Since the feedbacks are delayed by $d_{\max}$ at most, the regret caused by missing samples is also bounded. Let $|\mathcal{M}|$ denote the regret caused by missing samples.
		Then we have,
		\begin{align}
			\ln \frac{W_T}{W_0}&\leq -\eta\sum_t\sum_i p_t^i\tilde{l}_s^i\nonumber\\
			&~~~+\frac{\eta^2}{2}\sum_t\sum_i p_t^i(\tilde{l}_s^i)^2+|\mathcal{M}|.\label{equ:upper_bound}
		\end{align}
		Combining (\ref{equ:lower_bound}) and (\ref{equ:upper_bound}), we can get
		\begin{align*}
			-\eta\min_j\sum_{s=1}^{T}\tilde{l}_s^j-\ln K&\leq -\eta\sum_t\sum_i p_t^i\tilde{l}_s^i\nonumber\\
			&~~~+\frac{\eta^2}{2}\sum_t\sum_i p_t^i(\tilde{l}_s^i)^2+|\mathcal{M}|.
		\end{align*}
		Rearranging the above equation, we have 
		\begin{align}
			\sum_t\sum_i p_{t+d_t}^i\tilde{l}_t^i-\min_j\sum_t\tilde{l}_t^j&\leq \frac{1}{2}\eta\sum_t\sum_i p_{t+d_t}^i(\tilde{l}_t^i)^2\nonumber\\
			&+\frac{\ln K}{\eta}+|\mathcal{M}|.
			\label{equ:before_upper_bound}
		\end{align}
		Firstly, we need to figure out the relationship between $p_{t+d_t}^i$ and $p_t^i$.

		\newtheorem*{lemma1}{Lemma 1}
		\begin{lemma1}
			\label{lemma:1}
			Let $N_t=|\{s:s+d_s\in[t,t+d_t)\}|$ denote the amount of feedbacks that arrives between playing action and observing its feedback.
			For a subset of the integers $C$, corresponding to time steps, let 
			\begin{align}
			q^i(C)=\frac{\exp(-\eta\sum_{s\in C}\tilde{l}_s^i)}{\sum_j\exp(-\eta\sum_{s\in C}\tilde{l}_s^j)} \label{equ:q}.
			\end{align}
			Then the resulting probabilities fulfill for every time slot $t$ and action $i$,
			\begin{align}
			p_{t+d_t}^i-p_t^i\leq -\eta\sum_{m=1}^{N_t}q^i(C_{t-1}\cup\{z^\prime_n:n<m\})\tilde{l}_{z^\prime_i}^i, 
			\end{align}
			where $z^\prime_n$ is an enumeration of $\{s:s+d_s\in [t,t+d_t)\}$.		
						
		\end{lemma1}
		\begin{proof}[proof of Lemma 1]
			Refer to \emph{Lemma 10} of \cite{Thune2019}.
		\end{proof}
		By using \emph{Lemma 1} to bound the left hand side (LHS) of (\ref{equ:before_upper_bound}), we have
		\begin{align}
		\sum_t\sum_i p_{t}^i\tilde{l}_t^i&-\min_j\sum_t\tilde{l}_t^j\leq \frac{1}{2}\eta\sum_t\sum_i p_{t+d_t}^i(\tilde{l}_t^i)^2\nonumber\\		
		&+\eta\sum_t\sum_i\tilde{l}_t^i\sum_{m=1}^{N_t}q^i(C_{t-1}\cup\{z^\prime_n:n<m\})\tilde{l}_{z^\prime_n}^i\nonumber\\
		&+\frac{\ln K}{\eta}+|\mathcal{M}|.
		\label{equ:upper_bound1}
		\end{align}
		Next, we will bound the second term on the right hand side (RHS) of (\ref{equ:upper_bound1}).
		Since $t$ is not part of the enumeration $z_j^\prime$, so the two expectations are taken independently: $\mathbb{E}[\tilde{l}_{z_j^\prime}^i]\leq1$ and $\mathbb{E}[\tilde{l}_t^i]\leq1$. 
		Then, we have
		\begin{align*}
		E\Big[\sum_t\sum_i\tilde{l}_t^i\sum_{m=1}^{N_t}q^i(C_{t-1}\cup\{z_n^\prime:n<m\})\tilde{l}_{z_n^\prime}^i\Big]\leq \sum_t N_t.
		\end{align*}
		Since counting in how many intervals every loss is observed is the same as counting how many losses are observed in every interval.
		Thus, summing over $t$ or $s$ is equivalent, i.e.,
		\begin{align*}
		\sum_tN_t&\leq \sum_t|\{s:s+d_s\in[t,t+d_t)\}|\\
		&=\sum_s|\{t:s+d_s\in[t,t+d_t)\}|.
		\end{align*}
		Both $s$ and $t$ is bounded to $T$.
		Then we have
		\begin{align}
		&\sum_s|\{t:s+d_s\in[t,t+d_t)\}|\nonumber\\
		&\leq\sum_s|\{t>s:s+d_s\in[t,t+d_t)\}||\nonumber\\
		&~~~+|\{t<s:s+d_s\in[t,t+d_t)\}|. \label{equ:delay_num_1}
		\end{align}
		and first term  of (\ref{equ:delay_num_1}) can be bounded as
		\begin{align}
		&|\{t>s:s+d_s\in[t,t+d_t)\}|\nonumber\\
		&\leq|\{t>s:t\leq s+d_s\}\backslash{\{t>s:t+d_t<s+d_s\}}|\nonumber\\
		&\leq d_s-|\{t>s:t+d_t<s+d_s\}|,\label{equ:delay_num1}
		\end{align}
		The second term is can be bounded in a similar way, 
		\begin{align}
		&|\{t<s:s+d_s\in[t,t+d_t)\}|\nonumber\\
		&\leq|\{t<s:s+d_s<t+d_t\}|.\label{equ:delay_num2}
		\end{align}
		
		Then we note that by the prior equivalency of summing over $t$ or $s$, the last term in (\ref{equ:delay_num1}) cancel with (\ref{equ:delay_num2}) when summed, which bounds the second term of (\ref{equ:upper_bound1}) by $\eta D$.

		The next lemma can help us bound the first term on RHS of (\ref{equ:upper_bound1}).
		\newtheorem*{lemma2}{Lemma 2}
		\begin{lemma2}
			\label{lemma:2}
			The probabilities defined in (\ref{equ:q}) is 
			\begin{align}
			q^i(\Phi_i)\leq \Big(1+\frac{1}{2N-1}\Big)q^i(\Phi_{i-1})\label{lemma3},  \forall i, 
			\end{align}
		\end{lemma2}
		\begin{proof}[proof of Lemma 2]
			Refer to \emph{Lemma 11} of \cite{Thune2019}.
		\end{proof}
	
		Using \textbf{Lemma 2}, we have 
		\begin{align}
		&\sum_t\sum_ip_{t+d_t}^i(\tilde{l}_t^i)^2\nonumber\\
		&=\sum_t\sum_ip_t^i\frac{p_{t+d_t}^i}{p_t^i}(\tilde{l}_t^i)^2\nonumber\\
		&=\sum_t\sum_ip_t^i(\tilde{l}_t^i)^2\prod_{i=1}^{N_t}\frac{q^i(C_{t-1}\cup\{z_j^\prime:j\leq i\})}{q^i(C_{t-1}\cup\{z_j^\prime:j<i\})}\nonumber\\
		&\leq \sum_t\sum_ip_t^i(\tilde{l}_t^i)^2(1+\frac{1}{2N-1})^{N_t}\nonumber\\
		&\leq \sum_t\sum_ip_t^i(\tilde{l}_t^i)^2(1+\frac{1}{2N-1})^{2N-1}\nonumber\\
		&\leq \sum_t\sum_ip_t^i(\tilde{l}_a^i)^2e.
		\label{equ:prob_inequality}
		\end{align}

		Finally, the inequality of (\ref{equ:upper_bound1}) can be derived as 
		\begin{align}
		\label{equ:regret}
		\sum_t\sum_ip_{t}^i\tilde{l}_t^i-\min_j\sum_t\tilde{l}_t^j&\leq\frac{\ln K}{\eta}+\eta D\nonumber\\
		&+\frac{\eta}{2}\sum_t\sum_ip_t^i(\tilde{l}_t^i)^2e+|\mathcal{M}|.
		\end{align}
		Note that the loss definition in (\ref{equ:hat_l}) is a biased estimation loss.
		Next we use the following lemma to bound the unbiased estimation loss of the proposed algorithm.
		\newtheorem*{lemma6}{Lemma 3}
		\begin{lemma6}
			\label{lemma:6}
			With probability at least $1-\delta$,
			\begin{align}
				\sum_{t=1}^{T}(\tilde{l}_t^i-l_t^i)\leq \frac{\ln(K/\delta)}{2\beta},
			\end{align}
			simultaneously holds for all $i\in[K]$.
		\end{lemma6}
		\begin{proof}[proof of Lemma 3]
			Refer to \emph{Lemma 1} of \cite{neu2015}.
		\end{proof}

		Applying \emph{Lemma 3}, with probability at least $1-\theta$, we have
		\begin{align}
			&\sum_t\sum_ip_{t}^i\tilde{l}_t^i-\min_j\sum_t\tilde{l}_t^j\nonumber\\
			&\leq\frac{\ln K}{\eta}+\frac{\eta}{2}\sum_t\sum_ip_t^i(\tilde{l}_a^i)^2e+\eta D+|\mathcal{M}|.
			\label{equ:regret2}
		\end{align}
		From the definition of estimated loss in (\ref{equ:hat_l}), the first term of LHS in (\ref{equ:regret2}) can be converted to
		\begin{align}
			\sum_{i}p_t^i\tilde{l}_t^i&=\sum_{i=1}p_t^i\frac{l_t^i}{p_t^i+\beta l_t^i}\mathbbm{1}_{\{I_t=i\}}\nonumber\\
			&=\sum_{i}\mathbbm{1}_{\{I_t=i\}}l_t^i-\beta\sum_{i}\frac{(l_t^i)^2}{p_t^i+\beta l_t^i}\mathbbm{1}_{\{I_t=i\}}\nonumber\\
			&=\sum_{i}\mathbbm{1}_{\{I_t=i\}}l_t^i-\beta\sum_{i}l_t^i\tilde{l}_t^i.
			\label{equ:pl}
		\end{align}
		Similarly,
		\begin{align}
			\sum_ip_t^i(\tilde{l}_t^i)^2=\sum_ip_t^i\frac{l_t^i}{p_t^i+\beta l_t^i}\tilde{l}_t^i\leq\sum_il_t^i\tilde{l}_t^i.
			\label{equ:pl2}
		\end{align}
		Combining (\ref{equ:pl}), (\ref{equ:pl2}) with (\ref{equ:regret2}), we have
		\begin{align}
			&\sum_t(l_{t,I_t}-\beta\sum_il_t^i\tilde{l}_t^i)-\min\limits_j\sum_t(l_t^j+\frac{\ln(K/\delta)}{2\beta})\nonumber\\
			&\leq\frac{\ln K}{\eta}+\frac{\eta e}{2}\sum_t\sum_il_t^i\tilde{l}_a^i+\eta D+|\mathcal{M}|.\nonumber\\
			&\leq\frac{\ln K}{\eta}+\frac{\ln(K/\delta)}{2\beta}+\eta D+(\frac{\eta e}{2}+\beta)\sum_t\sum_il_t^i\tilde{l}_a^i+|\mathcal{M}|\nonumber\\
			&\leq\frac{\ln K}{\eta}+\frac{\ln(K/\delta)}{2\beta}+\eta D+(\frac{\eta e}{2}+\beta)\sum_t\sum_i\tilde{l}_t^i+|\mathcal{M}|\nonumber\\
			&\leq\frac{\ln K}{\eta}+\frac{\ln(K/\delta)}{2\beta}+\eta D+(\frac{\eta e}{2}+\beta)(KT+\frac{K\ln(K/\delta)}{2\beta})\nonumber\\
			&+|\mathcal{M}|.
			\label{equ:final_inequ}
		\end{align}
		By properly selecting $\eta=2\beta=\sqrt{\frac{\ln K+\ln(K/\delta)}{D+(e+1)KT/2}}$,
		we have
		\begin{align}
		\bar{\mathcal{R}}\leq&2\sqrt{\Big(D+\frac{(e+1)}{2}KT\Big)\Big(2\ln K-\ln\delta\Big)}\nonumber\\
		&+\frac{(e+1)}{2}K\ln(\frac{K}{\delta})+|\mathcal{M}|,
		\end{align}
		which finishes the proof of Theorem 1.		
	\end{proof}

	\section{}
	\label{proof:nash}
	\begin{proof}[proof of theorem 2]
	Let $\epsilon>0$, and define the ergodic average of the value of the game by 
	\begin{align}
		\bar{U}_T=\frac{\sum_{t=1}^{T}U(\bm{p}_t,\bm{q}_t)}{T}.
	\end{align}
	By utilizing the proposed algorithm with cost sequence $l_{t}^{i}=U(i,\bm{q}_t)$, we know from (\ref{equ:final_inequ}) that the row player guarantees that for any column strategy, in particular $\bm{q}_t$, and any row strategy $\bm{p}$, we have
	\begin{align}
		&E\Biggl\{\sum_{t=1}^{T}\Big(U(\bm{p}_t,\bm{q}_t)-U(\bm{p},\bm{q}_t)\Big)\Biggr\}\leq\frac{\ln K}{\eta}+\frac{\ln(K/\delta)}{2\beta}\nonumber\\
		&+\eta D+(\frac{\eta e}{2}+\beta)(KT+\frac{K\ln(K/\delta)}{2\beta})+|\mathcal{M}|.\label{equ:NE1}
	\end{align}
	Let $G$ denote the constant term of (\ref{equ:NE1}), such that
	\begin{align*}
		G = \frac{\ln K}{\eta}+\frac{\ln(K/\delta)}{2\beta}+(\frac{\eta e}{2}+\beta)\frac{K\ln(K/\delta)}{2\beta}+|\mathcal{M}|
	\end{align*}
	There exists $T_1>0$, 
	\begin{align}
		E\Biggl\{\bar{U}_T-U(\bm{p},\bm{\bar{q}}_T)\Biggr\}&=E\Biggl\{\frac{\sum_{t=1}^{T}\Big(U(\bm{p}_t,\bm{q}_t)-U(\bm{p},\bm{q}_t)\Big)}{T}\Biggr\}\nonumber\\
		&\leq\frac{G+(\frac{\eta e}{2}+\beta)KT+\eta D}{T}\nonumber\\
		&\leq \frac{\epsilon}{2}, ~~~~~~~~~~~~~\quad\quad\forall ~ T>T_1,
		\label{equ:row_ineq_2}
	\end{align}
	where the second inequality holds since $D$ is bounded by $d_{\max}$ and $G$ is negligible as $T\rightarrow\infty$.
	
	Then for the column player with cost sequence $l_{t}^{j}=1-U(\bm{p}_t,j)$, for every strategy $\bm{p}_t$ and strategy $\bm{q}$, we have
	\begin{align}
		E\Biggl\{\sum_{t=1}^{T}\Big(U(\bm{p}_t,\bm{q})-U(\bm{p}_t,\bm{q}_t)\Big)\Biggr\}\leq G+(\frac{\eta e}{2}+\beta)KT+\eta D.
		\label{equ:col_ineq}
	\end{align}	
	Existing $T_2>0$, 
	\begin{align}
		E\Biggl\{U(\bm{\bar{p}}_T,\bm{q})-\bar{U}_T\Biggr\}&=E\Biggl\{\frac{\sum_{t=1}^{T}\Big(U(\bm{p}_t,\bm{q})-U(\bm{p}_t,\bm{q}_t)\Big)}{T}\Biggr\}\nonumber\\
		&\leq\frac{G+(\frac{\eta e}{2}+\beta)KT+\eta D}{T}\nonumber\\
		&\leq \frac{\epsilon}{2}, \qquad\qquad \quad\quad    \forall~ T>T_2.
		\label{equ:col_ineq_2}
	\end{align}	
	Now, define $\bm{p}_T^b$ as the best response to $\bm{\bar{q}}_T$, and we have
	\begin{align}
		\bm{p}_T^b=\arg\min_{\bm{p}'}~U(\bm{p}',\bm{\bar{q}}_T)|\nonumber,
	\end{align}
	together with $\bm{q}_T^b$, the best response to $\bm{\bar{p}}_T$, and we have
	\begin{align}
		\bm{q}_T^b=\arg\max_{\bm{q}'}~U(\bm{\bar{p}}_T,\bm{q}')|\nonumber.
	\end{align}
	Choosing $\bm{p}=\bm{p}_T^b,\bm{q}=\bm{\bar{q}}_T$ in (\ref{equ:row_ineq_2}), (\ref{equ:col_ineq_2}) and adding them together, we have
	\begin{align}
		E\Biggl\{\Bigg|U(\bm{\bar{p}}_T,\bm{\bar{q}}_T)-\min_{\bm{p}'}~U(\bm{p}',\bm{\bar{q}}_T)\Bigg|\Biggr\}&=E\Bigl\{\bar{U}_T-U(\bm{p}_T^b,\bm{\bar{q}}_T)\Bigr\}\nonumber\\
		&+E\Biggl\{U(\bm{\bar{p}}_T,\bm{\bar{q}}_T)-\bar{U}_T\Biggr\}\nonumber\\
		&\leq\epsilon, ~\forall~ T>\max\{T_1,T_2\}, 
		\label{equ:nash1}
	\end{align}
	where the first equality holds since $U(\bm{\bar{p}}_T,\bm{\bar{q}}_T)\geq U(\bm{p}_T^b,\bar{\bm{q}}_T)$.
	Choosing $\bm{p}=\bm{\bar{p}}_T,\bm{q}=\bm{q}_T^b$ in (\ref{equ:row_ineq_2}), (\ref{equ:col_ineq_2}) and adding them together, we have
	\begin{align}
		E\Biggl\{\Bigg|U(\bm{\bar{p}}_T,\bm{\bar{q}}_T)-\max_{\bm{q}'}~U(\bm{\bar{p}}_T,\bm{q'})\Bigg|\Biggr\}&=E\Bigl\{\bar{U}_T-U(\bm{\bar{p}}_T,\bm{\bar{q}}_T)\Bigr\}\nonumber\\
		&+E\Biggl\{U(\bm{\bar{p}}_T,\bm{q}_T^b)-\bar{U}_T\Biggr\}\nonumber\\
		&\leq\epsilon,  \forall ~T>\max\{T_1,T_2\},
		\label{equ:nash2}
	\end{align}
	where the first equality holds as $U(\bm{\bar{p}}_T,\bm{\bar{q}}_T)\leq U(\bm{\bar{p}}_T,\bm{q}_T^b)$.
	Equations (\ref{equ:nash1}) and (\ref{equ:nash2}) show that $(\bm{\bar{p}}_T,\bm{\bar{q}}_T)$ can converge to $\mathcal{N}_{\epsilon}$ in the $L1$ sense.
	\end{proof}
	\balance
	\bibliographystyle{IEEEtran}
	\bibliography{refs}
	
\end{document}